\documentclass[conference]{IEEEtran}
\usepackage[left=0.63in,right=0.625in,top=0.75in,bottom=1in]{geometry}

\usepackage[T1]{fontenc}

\usepackage{cite}
\usepackage{graphicx}
\graphicspath{{Figure/}}
\usepackage[cmex10]{amsmath}
\usepackage{amsthm}
\usepackage{amssymb}
\usepackage{algorithmic}
\usepackage{array}
\usepackage{color}
\usepackage{epstopdf}
\usepackage{amsfonts}
\usepackage{siunitx}
\usepackage{soul}
\usepackage{xurl}

\usepackage[acronym,shortcuts]{glossaries}

\usepackage{pgfplots}
\pgfplotsset{compat=1.3}
\usepackage{tikz}							
\usetikzlibrary{shapes}
\usetikzlibrary{spy}
\usetikzlibrary{circuits}
\usetikzlibrary{arrows}	
\usepackage{subfig}

\usepackage{mathtools}

\newtheorem{theorem}{Theorem}
\newtheorem{proposition}{Proposition}

\makeglossaries
\newacronym[plural=BSs,firstplural=base stations (BSs)]{bs}{BS}{base station}
\newacronym[plural=PAs,firstplural=power amplifiers (PAs)]{pa}{PA}{power amplifier}
\newacronym{sdr}{SDR}{signal-to-distortion ratio}
\newacronym{snr}{SNR}{signal-to-noise ratio}
\newacronym{sndr}{SNDR}{signal-to-noise-and-distortion ratio}
\newacronym{los}{LoS}{line-of-sight}
\newacronym{z3ro}{Z3RO}{zero third-order distortion}
\newacronym{mimo}{MIMO}{multiple input multiple output}
\newacronym{mrt}{MRT}{maximum ratio transmission}
\newacronym{dpd}{DPD}{digital pre-distortion}

\newacronym{ula}{ULA}{uniform linear array}

\newacronym{psd}{PSD}{power spectral density}
\newacronym{aclr}{ACLR}{adjacent channel leakage ratio}
\newacronym{cpu}{CPU}{central processing unit}
\newacronym{iid}{i.i.d.}{independently and identically distributed}
\newacronym{ue}{UE}{user equipment}
\newacronym{nlos}{NLoS}{non-line-of-sight}
\newacronym{hpbm}{HPBM}{half power beam width}
\newacronym{rts}{RTS}{ray tracing simulator}
\newacronym{ag}{AG}{array gain}
\newacronym{arp}{ARP}{Antenna Reference Point}
\newacronym{ecdf}{eCDF}{Empirical Cumulative Distribution Function}
\newacronym{evm}{EVM}{Error Vector Magnitude}
\newacronym{fft}{FFT}{fast Fourier transform}
\newacronym{ib}{IB}{in-band}
\newacronym{im}{IM}{intermodulation}
\newacronym{mf}{MF}{matched filtering}
\newacronym{mr}{MR}{maximum ratio}
\newacronym{ofdm}{OFDM}{orthogonal frequency division multiplexing}
\newacronym{oob}{OOB}{out-of-band}
\newacronym{papr}{PAPR}{Peak-to-Average Power Ratio}
\newacronym{rx}{RX}{Receiver}
\newacronym{trp}{TRP}{Transmission Reception Point}
\newacronym{tx}{TX}{transmitter}
\newacronym{zf}{ZF}{zero forcing}
\newacronym{dl}{DL}{downlink}
\newacronym{rzf}{RZF}{Regularized Zero Forcing}
\newacronym{slc}{SLC}{Spatial Leakage Suppression}
\newacronym{kpi}{KPI}{Key Performance Indicator}
\newacronym{awgn}{AWGN}{Additive White Gaussian Noise}
\newacronym{qam}{QAM}{Quadrature Amplitude Modulation}

\begin{document}

\linespread{0.9}
\title{Z3RO Precoder Canceling Nonlinear Power Amplifier Distortion in Large Array Systems}

\author{\IEEEauthorblockN{Fran\c{c}ois Rottenberg,
		Gilles Callebaut,
		Liesbet Van der Perre
	}
	\IEEEauthorblockA{KU Leuven, ESAT-WaveCore, Ghent Technology Campus, 9000 Ghent, Belgium
	}
}

\maketitle

\begin{abstract}
	Large array-based transmission uses the combination of a massive number of antenna elements and clever precoder designs to achieve array gain and spatially multiplex different users. These precoders require linear front-ends, and more specifically linear \glspl{pa}. However, this reduces energy efficiency since \glspl{pa} are most efficient close to saturation, where they generate most nonlinear distortion. Moreover, the use of conventional precoders, such as \gls{mrt}, induces a coherent combining of distortion at the user location, degrading the signal quality. In this work, a linear precoder is proposed that allows working close to saturation while canceling the coherent combining of the third order nonlinear \gls{pa} distortion at the user location. In contrast to other solutions, the \gls{z3ro} precoder does not require prior knowledge of the signal statistics and the \gls{pa} model. The design consists of saturating a single or a few antennas on purpose together with an opposite phase shift to compensate for the distortion of all other antennas. The resulting array gain penalty becomes negligible as the number of base station antennas grows large.
\end{abstract}


\section{Introduction}\label{section:Introduction}

\glsresetall

\subsection{Problem Statement}

The \gls{pa} in wireless communications accounts for a major part of the energy consumption \cite{auer11}, and its operation requires a trade-off between linearity and energy efficiency. On the one hand, the linearity of the \gls{pa} is beneficial for the signal quality. On the other hand, the \gls{pa} has a maximal efficiency when its operating point is close to saturation, where its characteristic is nonlinear and creates distortion\cite{ENZ16}. Making a trade-off between these two aspects has always been a challenging problem\cite{lavr10}. Specifically in large array based transmission, such as massive MIMO in 5G, the combination of multicarrier transmission and precoding over a large number of antennas, leads to a high peak-to-average power ratio, requiring a large linear range of the \gls{pa} \cite{Fager2019}.

More specifically, in large array systems, \textit{i.e.}, one of the key technology enabler of 5G, authors have shown that the \gls{pa} nonlinear distortion is in general not uniformly radiated \cite{Larsson2018,Mollen2018TWC}. It is known that, when the \gls{bs} transmits in a dominant beamforming direction, the distortion is affected by the same array gain as the linearly amplified signal. This implies that nonlinear distortion can strongly limit the user performance, more specifically its \gls{sdr}, when working close to saturation and using conventional precoders.

\subsection{State-of-the-Art}

A first solution to improve the \gls{sdr} is to back-off power, \textit{i.e.}, to reduce the transmit power to let the \glspl{pa} work further away from saturation and thus in a more linear regime. However, this strongly degrades the \gls{pa} efficiency and reduces the transmit power, which can be problematic for users experiencing a large attenuation.

Another popular solution is to use \gls{dpd} compensation techniques \cite{cripps2006rf}. The idea is to pre-distort the signal that is at the input of the \gls{pa} to compensate for its nonlinear distortion. These techniques are data dependent and often require a feedback. Applying these techniques requires a certain complexity burden and power consumption, especially in large antenna systems, where it has to be implemented for each \gls{pa} \cite{Fager2019}. Moreover, even a perfect \gls{dpd} only linearizes input signal with an amplitude lower than the saturation level of the \gls{pa} (weakly nonlinear effects). Higher fluctuations are clipped due to saturation, resulting in nonlinear distortion (strongly nonlinear effects). This phenomenon is more likely as the back-off power is reduced, especially for high \gls{papr} signals such as multicarrier and massive \gls{mimo} systems.

As opposed to previous techniques, we propose a novel linear precoder, which compensates nonlinear distortion and can be implemented in addition to or instead of a \gls{dpd}. A recent work~\cite{aghdam2020distortion} has followed a similar approach. However, no simple closed-form solution for the precoder was presented. The contributions are properly formalized in the next subsection.


\subsection{Contributions}

The limitations of the \gls{mrt} precoder are put forward. This precoder is obtained by maximizing the \gls{snr}, without taking into account nonlinear distortion. As the \glspl{pa} enter saturation, distortion appears and \gls{mrt} is not optimal anymore. It even induces a coherent combining of distortion at the user locations. 

To address the shortcomings of the \gls{mrt} precoder, the \gls{z3ro} precoder is proposed. It is obtained by maximizing the \gls{snr} under the constraint that the third-order distortion is null at the user location. The precoder has, as \gls{mrt}, a low complexity. Moreover, it does not depend on the \gls{pa} parameters and the transmit signal statistics, which makes it simple to compute and to implement. It is obtained by saturating a given subset of antennas and applying an opposite phase shift to them. While using a single saturated antenna is optimal in terms of \gls{snr}, using a larger number of saturated antennas results in a more spatially focused distortion pattern, which is useful regarding unintended locations and total radiated distortion. Moreover, the \gls{snr} penalty of this precoder as compared to \gls{mrt} becomes negligible in the large antenna case. All code used to generate figures are in open-source at \url{github.com/DRAMCO/z3ro-precoder-single-user}.

{\textbf{Notations}}: Superscript $(.)^*$ stands for conjugate operator. $\jmath$ is the imaginary unit. The symbols $\mathbb{E}(.)$, $\angle(.)$, $\Im(.)$ and $\Re(.)$ denote the expectation, phase, imaginary and real parts, respectively.



\section{System Model}\label{section:transmission_model}


\subsection{Signal Model}


We consider a a large array-based system with a single user and single \gls{bs} equipped with $M$ antennas.  The complex symbol intended for the user is denoted by $s$, with variance $p$. The signal $s$ is precoded at transmit antenna $m$ using a precoder coefficient $w_{m}$. The complex baseband representation of the signal before the \gls{pa} of the corresponding antenna is denoted by $x_{m}$ and is given by $x_{m}= w_{m} s$.


\subsection{Power Amplifier Model}

In the following, all \glspl{pa} are assumed memoryless and to have the same transfer function. For the sake of clarity and without loss of generality, the linear gain of the \gls{pa} is set to one. We only consider the third order nonlinear distortion of the \gls{pa}. This approximation regime is valid as the \gls{pa} enters saturation regime, which creates nonlinear distortion but not enough for higher order terms to provide a significant contribution. Under these assumptions, the \gls{pa} output of antenna $m$ can be written as
\begin{align}
y_{m}= x_{m} +a_3 x_{m}  |x_{m}|^{2}, \label{eq:y_mn_gen}
\end{align}
where the coefficient $a_3$ characterizes the nonlinear characteristic of the \gls{pa}, including both amplitude-to-amplitude modulation (AM/AM) and amplitude-to-phase modulation (AM/PM).

\subsection{Channel Model}\label{subsection:channel_model}

The complex channel gain from antenna $m$ to the user is denoted by $h_{m}$. 
The received signal is given by
\begin{align}
r&=\sum_{m=0}^{M-1} h_{m} y_m + v, \label{eq:r_SU}
\end{align}
where $v$ is zero mean circularly symmetric complex Gaussian noise with variance $\sigma_v^2$. In the following, at some places, a pure \gls{los} channel to each user will be considered. In this particular case, $h_{m}$ can be written as
\begin{align}
h_{m}= \sqrt{\beta} e^{-\jmath \phi_{m}} \label{eq:channel_model_LOS}
\end{align}
The real positive coefficient $\beta$ models the path loss. For a narrowband system, the difference of propagation distance between each of the antennas and the user results in an antenna-dependent phase shift $\phi_{m}$, which can be directly related to the antenna location and the angular direction of the user. For a \gls{ula}, the phase shift is given by $\phi_{m}=m\frac{2\pi}{\lambda_c} d\cos(\theta)$, where $\lambda_c$ is the carrier wavelength, $d$ the inter-antenna spacing and $\theta$ is the user angle. 

The radiation pattern in an arbitrary direction $\tilde{\theta}$ can be computed as
\begin{align}
P(\tilde{\theta})&=\mathbb{E}\left(\left|\sum_{m=0}^{M-1}y_me^{-\jmath\tilde{\phi}_m}\right|^2\right)
,\label{eq:radiation_pattern}
\end{align}
where $\tilde{\phi}_m=m \frac{2\pi}{\lambda_c} d\cos(\tilde{\theta})$. Defining the total transmit power $P_T=\int_{-\pi}^{\pi}P(\tilde{\theta})d\tilde{\theta}$, the array directivity is $D(\tilde{\theta})=\frac{P(\tilde{\theta})}{P_T/2\pi}$, \textit{i.e.}, $P(\tilde{\theta})$ is normalized with respect to an isotropic radiator. 


\section{Limitations of the Maximum Ratio Transmission Precoder}
\label{section:SU}

\begin{figure*}[t!]
	\centering 
	\subfloat[MRT, $M=32$.
	]{
		\resizebox{0.4\linewidth}{!}{%
			{\includegraphics[clip, trim=2.5cm 2.6cm 2.5cm 3cm, scale=1]{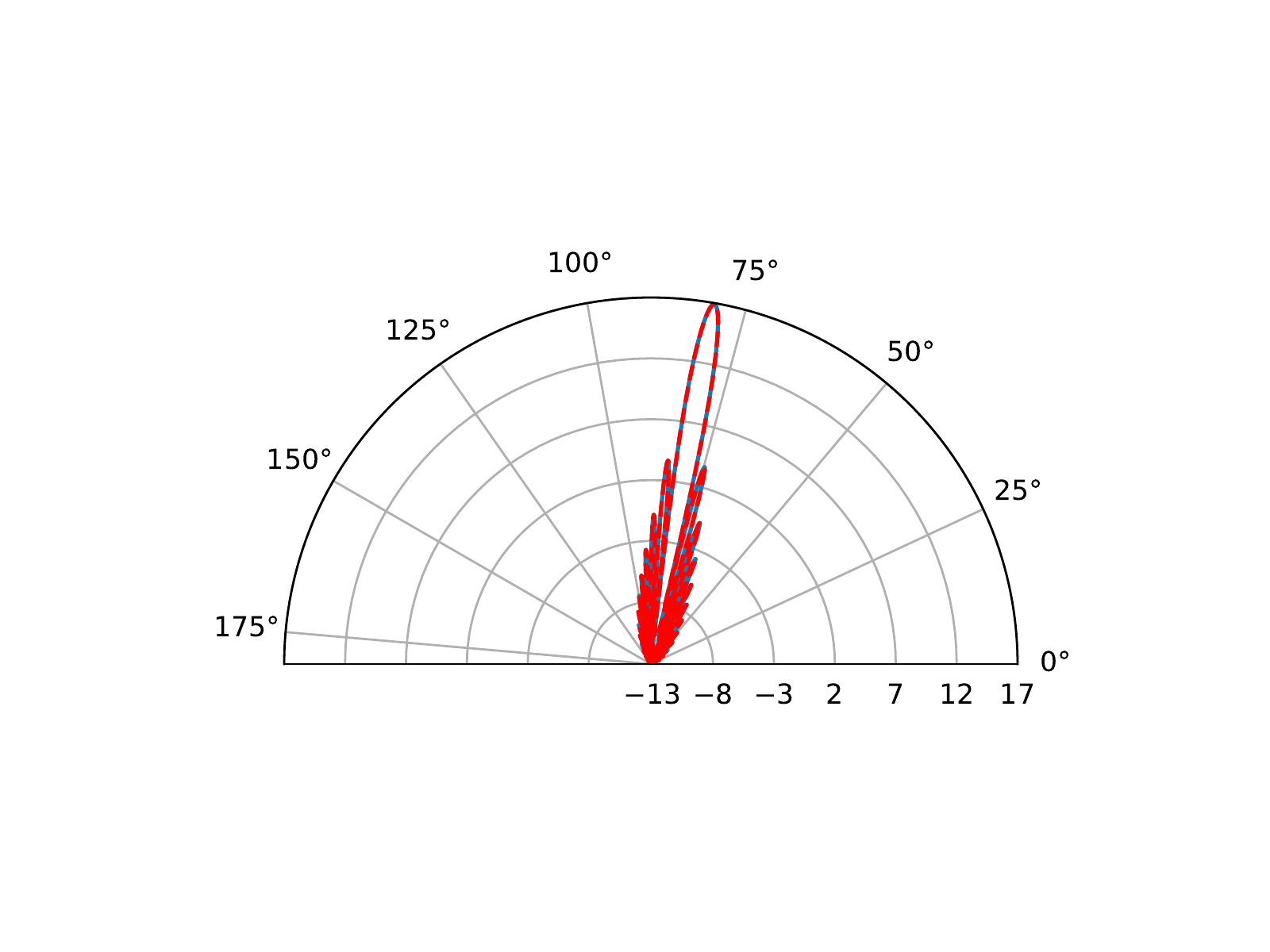}} 
		}
	} 
	\subfloat[\gls{z3ro}, $M=2$.]{
		\resizebox{0.4\linewidth}{!}{%
			{\includegraphics[clip, trim=2.5cm 2.6cm 2.5cm 3cm, scale=1]{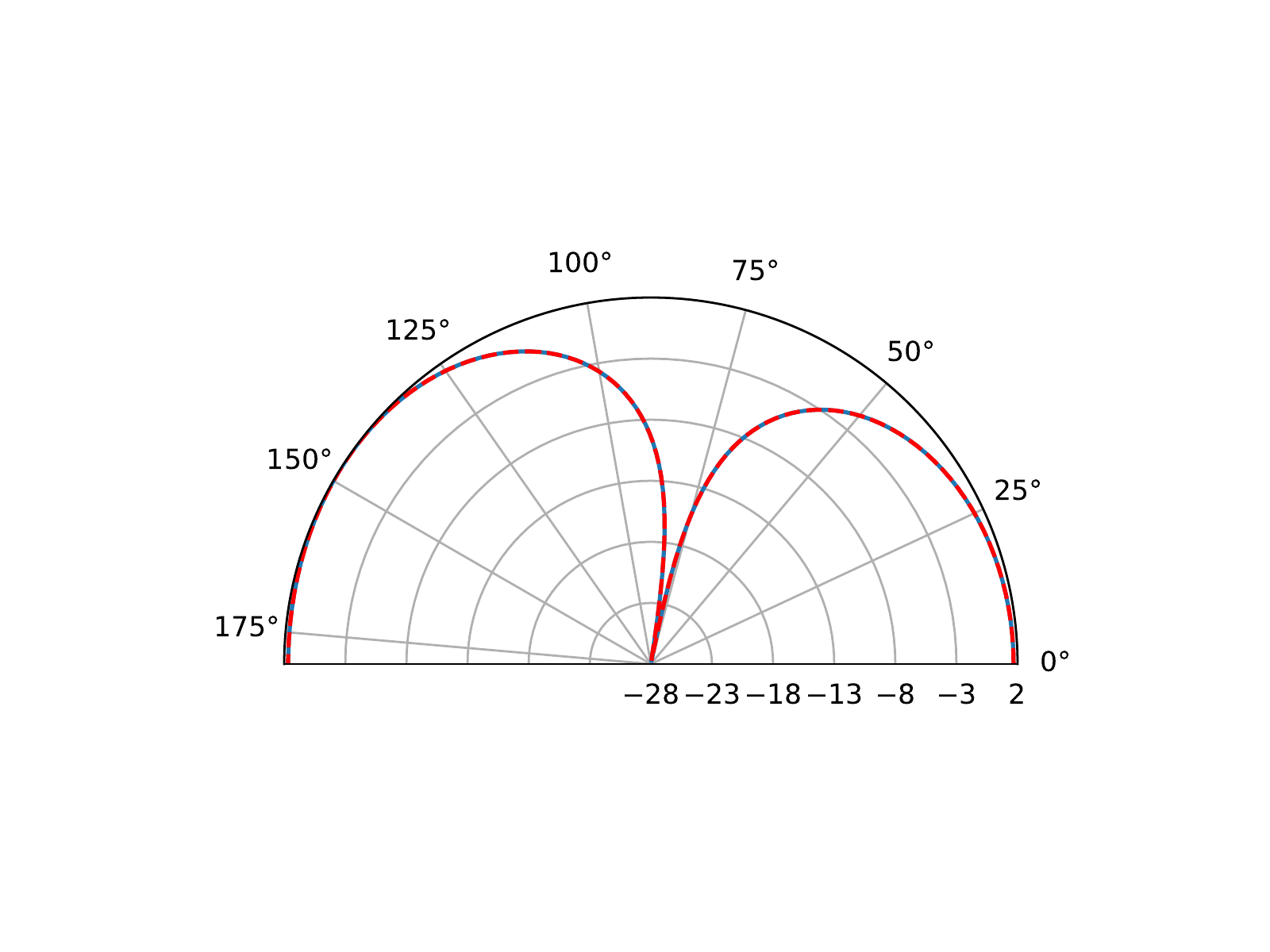}} 
		}
	}\\\vspace{-1em}
	\subfloat[\gls{z3ro}, $M=8$.]{
		\resizebox{0.4\linewidth}{!}{%
			{\includegraphics[clip, trim=2.5cm 2.6cm 2.5cm 3cm, scale=1]{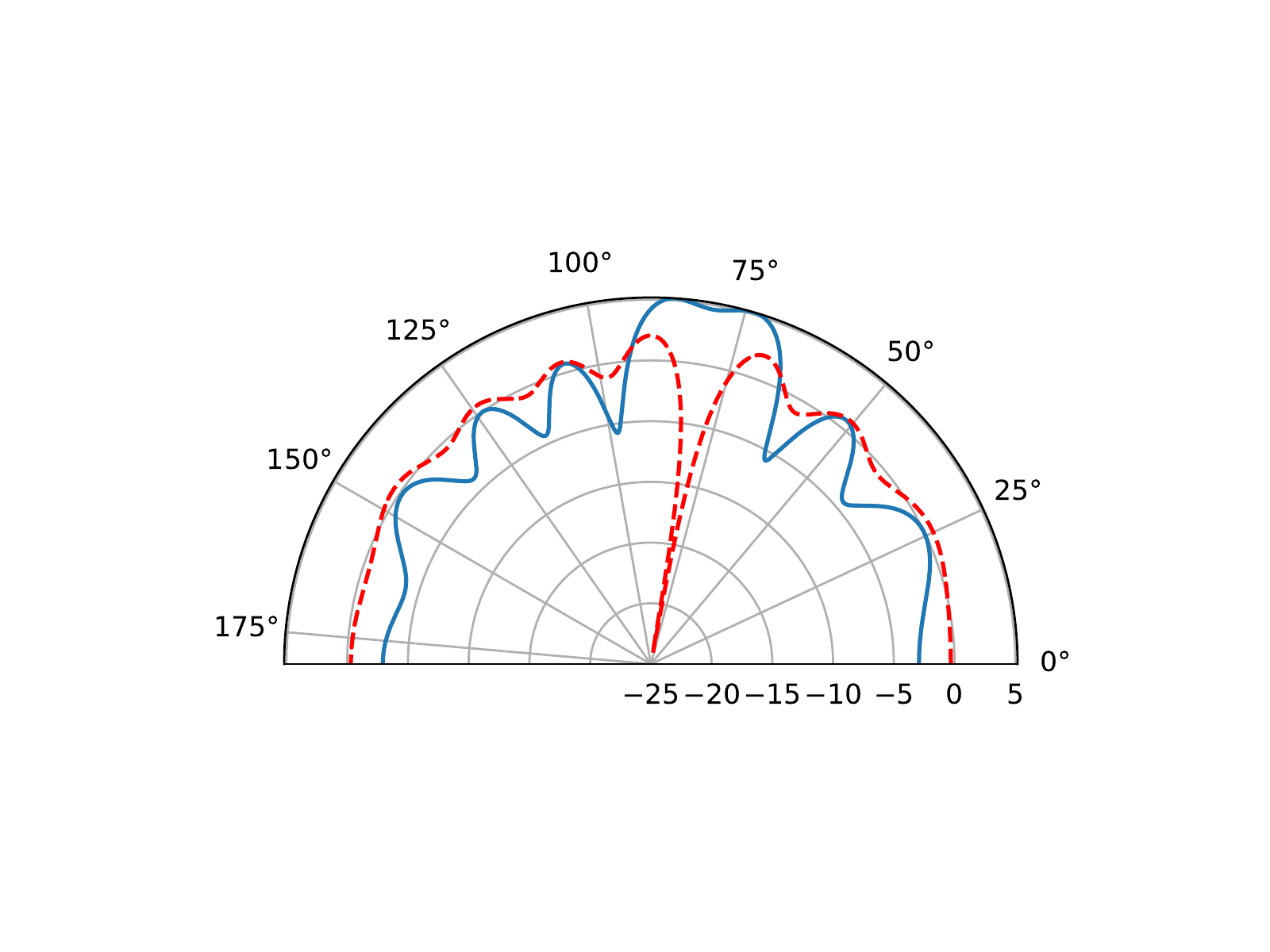}} 
		}
	}
	\subfloat[\gls{z3ro}, $M=32$.]{
		\resizebox{0.4\linewidth}{!}{%
			{\includegraphics[clip, trim=2.5cm 2.6cm 2.5cm 3cm, scale=1]{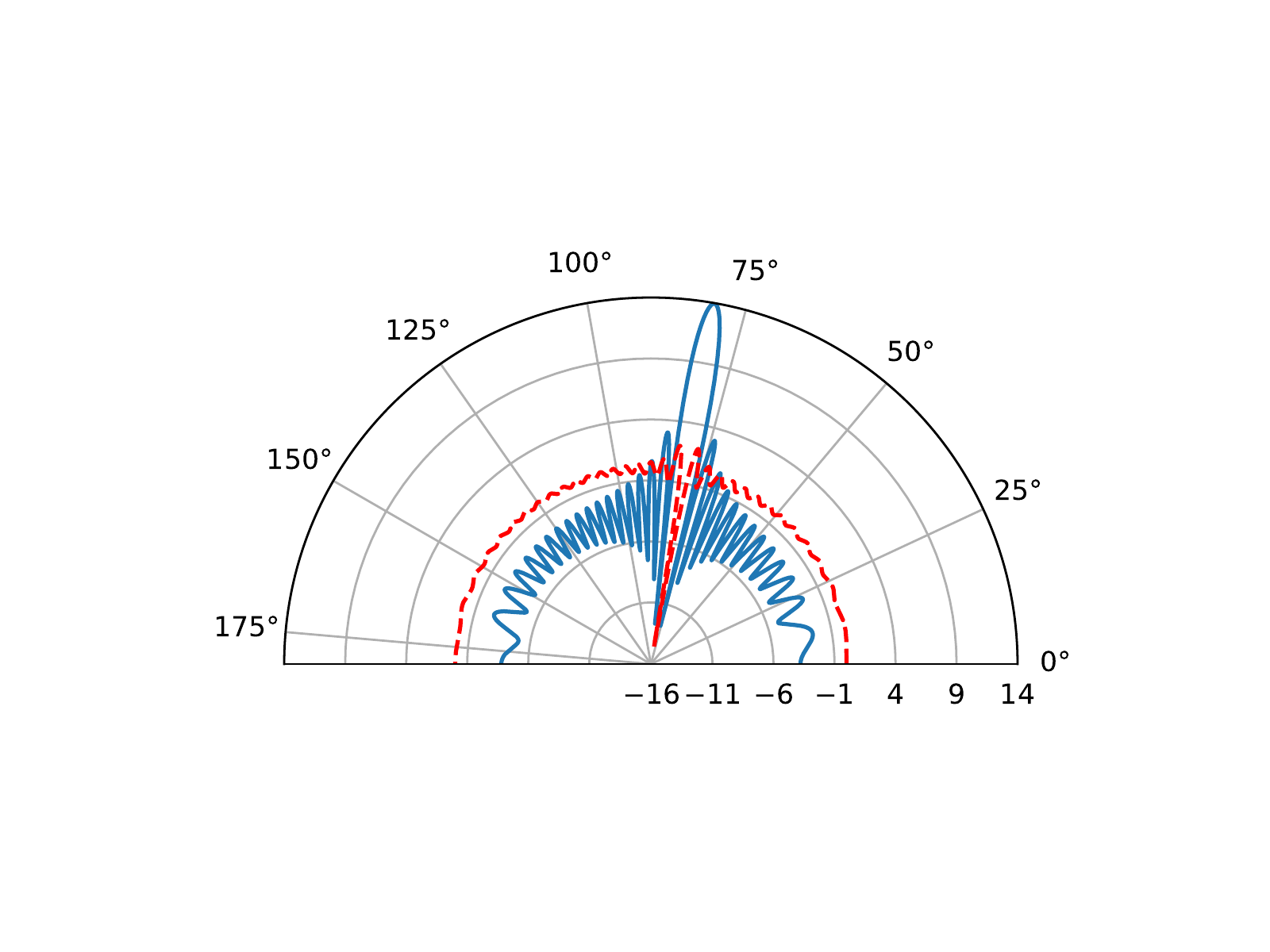}} 
		}
	}
	\caption{Directivity pattern [dB] of the signal (continuous blue) and third-order distortion (dashed red) for a pure \gls{los} channel and half-wavelength \gls{ula}. User angle is $\theta=80^{\circ}$ and $M_s=1$ saturated antenna. For \gls{mrt}, distortion coherently combines in the user direction while it is null for \gls{z3ro}. The array gain penalty is infinite for $M=2$ while it vanishes as $M$ grows large.}
	\label{fig:directivity_patterns} 
		\vspace{-1em}
\end{figure*}

The well known \gls{mrt} precoder is obtained by maximizing the received \gls{snr}, disregarding the nonlinear distortion terms at the output of the \gls{pa}, \textit{i.e.}, approximating $a_3\approx0$. Under this approximation, the output of the \gls{pa} becomes $y_m \approx x_m$ and the received signal becomes
\begin{align*}
r&\approx \sum_{m=0}^{M-1} h_{m} x_m + v=s \sum_{m=0}^{M-1} h_{m} w_m + v.
\end{align*}
The received \gls{snr} is given by
\begin{align}
\text{SNR}=\frac{p}{\sigma_v^2}  \left|\sum_{m=0}^{M-1} h_m w_m\right|^2. \label{eq:SNR_SU}
\end{align}
The \gls{mrt} precoder is found by optimizing the \gls{snr} under a total transmit power constraint
\begin{align*}
\max_{w_0,...,w_{M-1}} \text{SNR} \quad \text{s.t.}\quad p\sum_{m=0}^{M-1}|w_m|^2=pM.
\end{align*}
The problem can be solved using, \textit{e.g.}, the Lagrangian multipliers technique, giving the \gls{mrt} precoder
\begin{align*}
w_m^{\mathrm{MRT}}&=\alpha h_m^*,\ \alpha =\sqrt{\frac{M}{\sum_{m'=0}^{M-1}|h_{m'}|^2}},\\ \text{SNR}_{\mathrm{MRT}}&= pM \frac{\sum_{m'=0}^{M-1}|h_{m'}|^2}{\sigma_v^2},
\end{align*}
where $\alpha $ is a normalization constant\footnote{Note that the solution is equivalent up to a constant phasor applied to each antenna element. It is set to one here for simplicity.}. Assuming a unit per-antenna channel variance $\mathbb{E}(|h_{m'}|^2)=1$, the \gls{mrt} achieves an array gain of a factor $M$ with respect to the transmit power $pM$. The \gls{mrt} precoder is optimal as long as the \gls{pa} works in its linear regime. As $p$ increases, nonlinear terms will be amplified and distortion becomes non-negligible. The \gls{pa} output $y_m$ given in (\ref{eq:y_mn_gen}) can be evaluated for $x_m=w_m^{\mathrm{MRT}} s$
\begin{align*}
y_{m}&= s \alpha h_{m}^*   + a_3 s  |s|^{2} \alpha ^3 h_{m}^*  |h_{m}|^{2} 
\end{align*}
and the received signal (\ref{eq:r_SU}) becomes
\begin{align*}
r= s  \alpha \sum_{m=0}^{M-1}|h_{m}|^{2} +a_3s  |s|^{2}  \alpha ^3 \sum_{m=0}^{M-1}|h_{m}|^{4} + v,
\end{align*}
where it can be seen that the channel coherently combines both the linear term and the nonlinear term, \textit{i.e.}, the phases are matched. As a result, distortion coherently adds up at the user location and becomes the limiting factor at high power.

To clarify this, let us consider a pure \gls{los} channel (\ref{eq:channel_model_LOS}) giving $h_m=\sqrt{\beta} e^{-\jmath \phi_m}$ and $\alpha =1/\sqrt{\beta}$. Then,
\begin{align*}
y_{m}&= e^{\jmath \phi_m} (s    + a_3 s  |s|^{2}) \\
r&= \sqrt{\beta} M (s+ a_3s  |s|^{2})+ v,
\end{align*}
where it is clear that the array gain $M$ affects both linear and nonlinear terms. Moreover, one can see that both the linear and nonlinear terms are beamformed in the same direction, as they are both affected by the term $e^{\jmath \phi_m}$. The radiation pattern can be evaluated using (\ref{eq:radiation_pattern}). An example of the directivity pattern of linear/nonlinear terms for \gls{mrt} precoding and a ULA is shown in Fig.~\ref{fig:directivity_patterns}~(a). It illustrates that the directivity pattern of both linear and nonlinear terms is exactly the same.


\section{Zero Third-Order Distortion Precoder}
\label{subsection:zero_third_order_distortion}

The \gls{mrt} precoder induces a coherent combining of distortion at the user location, as demonstrated in the previous section. As $p$ increases, the \glspl{pa} will be more saturated and the user performance will be limited by its \gls{sdr}. Since linear and nonlinear terms are affected by the same array gain, increasing the number of antennas does not help to improve the \gls{sdr}. This section presents the \gls{z3ro} precoder which is designed to maximize the linear array gain to boost the \gls{snr} at the user location while canceling the third order distortion.

The received signal at the user location for a general linear precoder $w_m$ is
\begin{align*}
r &= \sum_{m=0}^{M-1} h_m  x_m+a_3 \sum_{m=0}^{M-1} h_m x_m |x_m|^2 + v\\
&= s \sum_{m=0}^{M-1} h_m w_m +a_3 s |s|^2 \sum_{m=0}^{M-1} h_m w_m |w_m|^2 + v.
\end{align*}
The distortion term can be forced to zero by ensuring that
\begin{align}
\sum_{m=0}^{M-1} h_m w_m |w_m|^2&=0. \label{eq:zero_distortion_constraint}
\end{align}
Note that this constraint does not depend on the statistics of the transmit symbols $s$ and the \gls{pa} parameter $a_3$, which makes it practical to implement. A similar condition was obtained in \cite{aghdam2020distortion}. However, the authors made the pessimistic conclusion that considering this constraint leads to a considerable reduction of array gain. Indeed, take the two antenna case $M=2$ and a \gls{los} channel $h_m=\sqrt{\beta} e^{-\jmath \phi_m}$. If the user angle is coming from broadside, it implies that $\phi_0=\phi_1=0$ and the constraint (\ref{eq:zero_distortion_constraint}) implies that
\begin{align*}
w_0|w_0|^2&=-w_1|w_1|^2 \\
|w_0|^3 e^{\jmath \angle w_0}&=-|w_1|^3 e^{\jmath \angle w_1},
\end{align*}
which implies that $w_0$ and $w_1$ should have the same magnitude but opposite phase, \textit{e.g.}, $w_0=-w_1$, which leads to a zero array gain, \textit{i.e.}, $|w_0+w_1|^2=0$. The same result occurs for any user angle, as depicted in Fig.~\ref{fig:directivity_patterns}~(b). However, it is shown further that, with the proposed \gls{z3ro} design, as the number of antennas $M$ grows large, the loss in array gain becomes negligible, as depicted in Fig.~\ref{fig:directivity_patterns}~(c) and (d). In the following, we first derive the expression of the \gls{z3ro} precoder in the \gls{los} case, before extending it to general channels.

\subsection{Line-of-Sight Channel}

For a \gls{los} channel $h_m=\sqrt{\beta} e^{-\jmath \phi_m}$, under constraint (\ref{eq:zero_distortion_constraint}), the precoder optimization problem can be formulated as
\begin{align}
\max_{w_0,...,w_{M-1}} \text{SNR}=\frac{\beta p}{\sigma_v^2}  \left|\sum_{m=0}^{M-1} e^{-\jmath \phi_m} w_m\right|^2, \label{eq:3rd_problem}
\end{align}
under the two constraints
\begin{align}
\text{Transmit power: }&p\sum_{m=0}^{M-1}|w_m|^2=pM, \label{eq:transmit_power_constraint}\\  
\text{Zero third-order distortion: }&\sum_{m=0}^{M-1} e^{-\jmath \phi_m} w_m |w_m|^2=0.
\end{align}
Note that (\ref{eq:transmit_power_constraint}) is not exactly a transmit power constraint, but rather the total power at the input of the \glspl{pa}. Indeed, the nonlinear transmited power is disregarded for simplicity and for the sake of comparison with the \gls{mrt} precoder, which is found under a similar constraint.

The above problem is non convex and not trivial to solve. However, it can be first reformulated in a simpler form using the change of variable $g_m=w_me^{-\jmath \phi_m}$. Then, without loss of generality, the optimization problem can be solved as a function of $g_m$ instead of $w_m$. The reformulated problem becomes
\begin{align*}
\max_{g_0,...,g_{M-1}} \left|\sum_{m=0}^{M-1} g_m\right|^2\ \text{s.t.}\ \sum_{m=0}^{M-1} |g_m|^2=M,\ \sum_{m=0}^{M-1} g_m |g_m|^2=0.
\end{align*}
For an optimal solution $g_m$, the corresponding $w_m$ is given by $w_m=g_me^{\jmath \phi_m}$. Moreover, from the above formulation, a conjecture can be made: an optimal $g_m$ should be purely real up to a constant phasor.\footnote{Unfortunately, we could not demonstrate this conjecture yet even though it was found to be valid for extensive numerical simulations.} If this phasor is set to one, the problem is converted to an all real problem
\begin{align}
\max_{g_0,...,g_{M-1}} \left(\sum_{m=0}^{M-1} g_m\right)^2\ \text{s.t.}\ \sum_{m=0}^{M-1} g_m^2=M,\ \sum_{m=0}^{M-1} g_m^3=0. \label{eq:all_real_3rd_problem}
\end{align}
The following theorem solves the problem (\ref{eq:all_real_3rd_problem}) and gives the expression for the optimum precoder $w_m=g_me^{\jmath \phi_m}$.

\begin{theorem}\label{theorem:3rd_order} Critical points of problem (\ref{eq:all_real_3rd_problem}), providing a non zero array gain, are obtained by using a number of antennas $M_s$ with an opposite phase shift and saturated in such a way that they compensate for the distortion due to all other antenna elements. This leads to the \gls{z3ro} precoder design:
	\begin{align*}
	w_m^{\mathrm{Z3RO},M_s}&=\alpha e^{\jmath \phi_m}\begin{cases}
	- \left(\frac{M-M_s}{M_s}\right)^{1/3}\ &\text{if}\ m=0,...,M_s-1\\
	1\ &\text{otherwise}
	\end{cases},
	\end{align*}
	where $\alpha={\sqrt{M}}/{\sqrt{M-M_s+M_s^{1/3}(M-M_s)^{2/3}}}$ is a power normalization constant. The $M_s$ elements can be chosen arbitrarily, with $M/2>M_s>0$. The global maximum of (\ref{eq:all_real_3rd_problem}) is obtained by using a single saturated antenna, \textit{i.e.} $M_s=1$. The received \gls{snr} is
	\begin{align}
	\text{SNR}&=\frac{\beta p}{\sigma_v^2}M\frac{((M-M_s)^{2/3}-M_s^{2/3})^2}{(M-M_s)^{1/3}+M_s^{1/3}}. \label{eq:SNR_Z3RO}
	\end{align}
	For a fixed $M_s$, $\lim_{M\rightarrow +\infty} \text{SNR}/\text{SNR}_{\mathrm{MRT}}= 1$, implying that the reduction of array gain becomes negligible while the third distortion order is null.
	
\end{theorem}
\begin{proof}
	See appendix.
\end{proof}

The theorem provides a precoder which is simple to compute and to implement. Its complexity is approximately equal to the one of the \gls{mrt} precoder. As explained earlier, it does not require any knowledge about the \gls{pa} or the statistics of the transmit signal. The global optimum is obtained by using a single saturated antenna. However, other critical points provide useful trade-offs, as explained further below.

\begin{figure}[t!]
	\centering 
	\resizebox{0.9\linewidth}{!}{%
		{\input{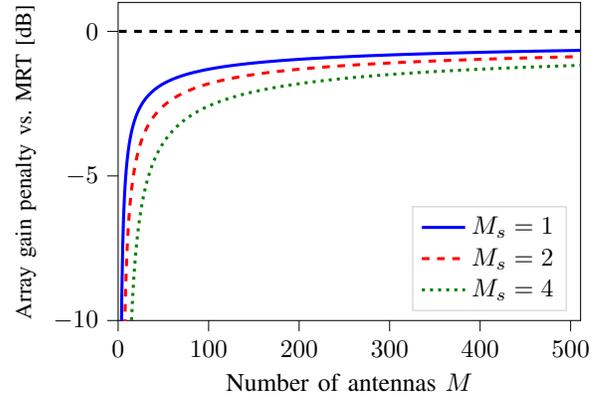}}
	} 
	\caption{As the number of antennas increases, the penalty in array gain of the \gls{z3ro} precoder vanishes, as compared to \gls{mrt}. $M_s$ is the number of saturated antennas.}
	\label{fig:array_gain} 
\end{figure}

\textbf{Array gain}: The \gls{z3ro} precoder is proposed for large array systems, operating in the saturation regime. In this regime, the \gls{sdr} is limiting rather than the \gls{snr}, implying that the reduction of array gain with respect to \gls{mrt} becomes negligible. Moreover, as shown in Fig.~\ref{fig:array_gain}, this array gain penalty vanishes for large array systems (as $M$ grows large). As an example, for $M_s=1$ and $M=64$, the \gls{mrt} precoder achieves an array gain of about 18 dB versus 16.5 dB for the proposed precoder, while the third distortion order is completely canceled. 


\begin{figure}[t!]
	\centering 
	\subfloat[Signal radiation pattern (dB), $M=32$.]{
		\resizebox{0.8\linewidth}{!}{%
			{\includegraphics[clip, trim=2.5cm 2cm 2.5cm 3cm, scale=1]{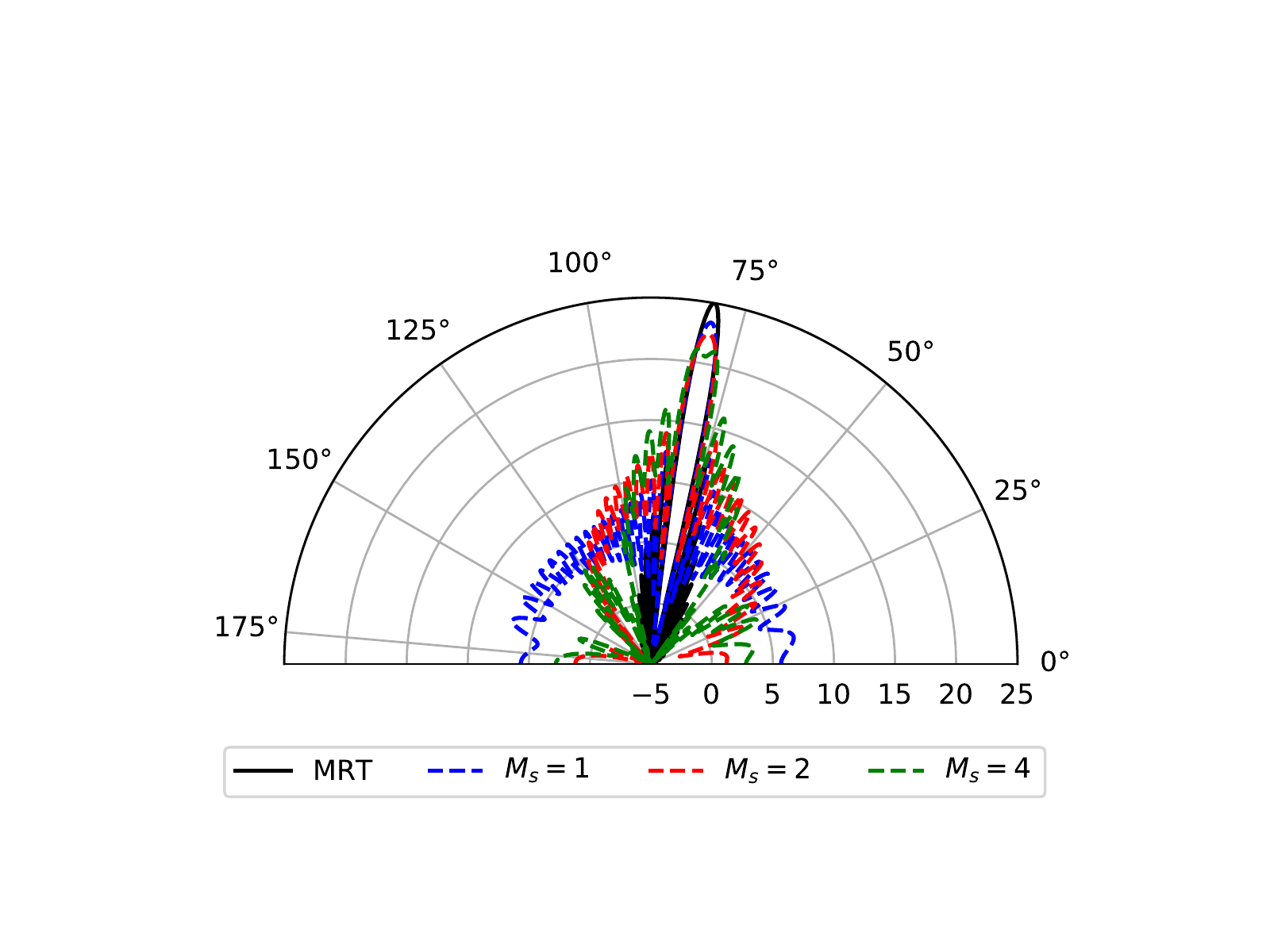}}
		}
	}
	
	\subfloat[Third-order distortion radiation pattern (dB), $M=32$.
	]{
		\resizebox{0.8\linewidth}{!}{%
			{\includegraphics[clip, trim=2.5cm 2cm 2.5cm 3cm, scale=1]{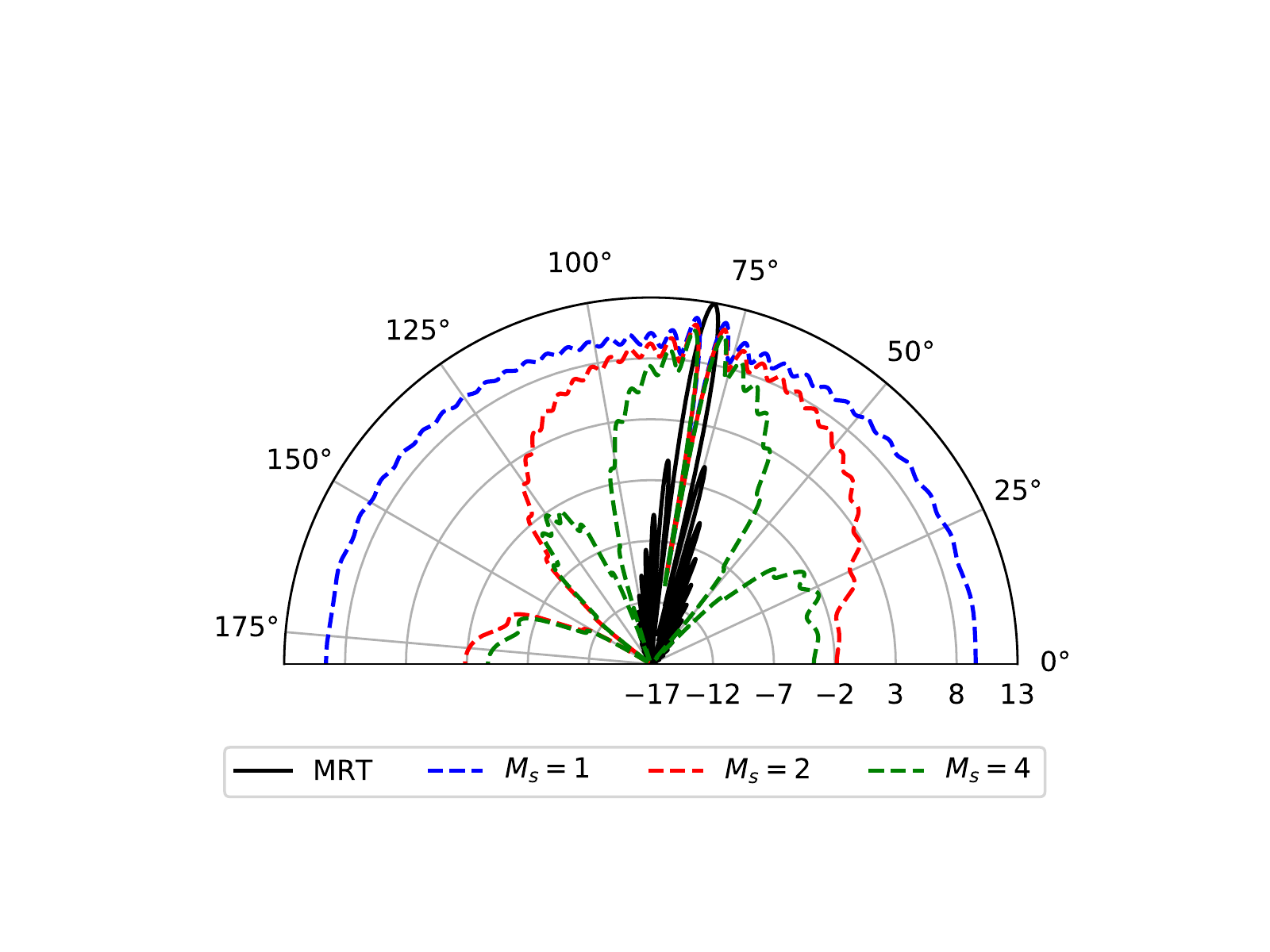}}
		}
	} 
	
	\caption{Radiation pattern for \gls{los} channel and half-wavelength \gls{ula}. As more antennas become saturated ($M_s\nearrow$), the array gain of the \gls{z3ro} precoder decreases. On the other hand, the total radiated power decreases and it becomes more spatially focused, which is beneficial regarding unintended directions.}
	\label{fig:radiation_patterns} 
		\vspace{-1em}
\end{figure}

\textbf{Distortion radiation pattern}: While Fig.~\ref{fig:directivity_patterns} shows directivity patterns of the signal and distortion for the \gls{mrt} and the proposed precoders, Fig.~\ref{fig:radiation_patterns} shows their absolute radiation pattern for different values of $M_s$. In Fig.~\ref{fig:radiation_patterns}~(a), in accordance with the previous discussions, the array gain decreases when $M_s$ increases. This would imply only using the design $M_s=1$. However, as shown in Fig.~\ref{fig:radiation_patterns}~(b), this design, as compared to \gls{mrt}, leads to an increase of total distortion power, which is mainly radiated towards unintended locations. This leads to interference for potential observers, especially since \gls{pa} nonlinearities induce out-of-band emissions. Hence, a user in an adjacent band could suffer from it. On the other hand, as $M_s$ increases, the distortion becomes focused ``approximately in the user direction, except for the exact user direction, where it is null by design''. Moreover, the total radiated distortion power is reduced. This comes from the fact that distortion is beamformed and benefits from an array gain. As a result, choosing the value $M_s$ offers a trade-off between array gain and spatial focusing of the distortion.

\subsection{General Channel}

To extend the precoder design to a general channel $h_m$, a heuristic design is proposed, directly inspired by the previous one, and satisfying the zero third-order distortion constraint (\ref{eq:zero_distortion_constraint}). 
\begin{proposition}
	An extension of the \gls{z3ro} precoder to a general channel is obtained as
	\begin{align*}
	w_m^{\mathrm{Z3RO},M_s}	&=\alpha h_m^*\begin{cases}
	-\gamma &\text{if}\ m=0,...,M_s-1\\
	1 &\text{otherwise}
	\end{cases},
	\end{align*}
	where $\gamma$ is the additional gain of the saturated antennas
	\begin{align*}
	\gamma&=\left(\frac{\sum_{m'=M_s}^{M-1} |h_{m'}|^4}{\sum_{m''=0}^{M_s-1} |h_{m''}|^4}\right)^{1/3}
	\end{align*}
	and $\alpha$ is a power normalization constraint given by
	\begin{align*}
	\alpha=\frac{\sqrt{M}}{\sqrt{\sum_{m'=M_s}^{M-1} |h_{m'}|^2+\gamma^2\sum_{m=0}^{M_s-1} |h_{m}|^2}}.
	\end{align*}	
	If the channel gains $h_m$ are assumed to be independent and identically distributed (i.i.d.) and so that $\mathbb{E}(|h_m|^2)=\beta$, as $M$ and $M_s$ grow large, the \gls{z3ro} precoder asymptotically achieves the same \gls{snr} performance as in Theorem~\ref{theorem:3rd_order}.
	
\end{proposition}
\begin{proof}
	One can check that the zero third-order distortion constraint (\ref{eq:zero_distortion_constraint}) is satisfied for the proposed precoder. To show that the precoder achieves the same \gls{snr} as given in (\ref{eq:SNR_Z3RO}), one can use the law of large number to show that, as $M\rightarrow +\infty$ and $M_s\rightarrow +\infty$, each sample average of i.i.d. elements converges to their expectation. 
\end{proof}

The \gls{z3ro} precoder proposed in Theorem~\ref{theorem:3rd_order} is found as a special case of a \gls{los} channel. Hence, the same notation $w_m^{\mathrm{Z3RO},M_s}$ was used. Furthermore, it shows that, for i.i.d. channels and a similar averaged channel power, the same \gls{snr} performance as in \gls{los} can be found, provided that the number of antennas $M$ and saturated antennas $M_s$ is large enough. An additional advantage of using a larger $M_s$ here is that it provides some diversity. This helps to prevent the case where saturated antennas would have a small channel gain. Indeed, a pathological case would be to use a single saturated antenna with low channel gain $|h_m|$, so that $\gamma$ would be very large and would consume more power. Using more saturated antennas allows to avoid these fades.


\section{Simulation Results \& Discussion}
\label{section:simulation_results}

Let us compare the \gls{mrt} and the \gls{z3ro} precoders. The Bussgang theorem \cite{bussgang1952crosscorrelation,demir2020bussgang} implies that the received signal can be decomposed as $r=G s + d + v$, where $d$ is the nonlinear distortion, which is uncorrelated with the transmit signal $s$ and noise $v$. The linear gain $G$ can be evaluated as $G=\mathbb{E}(rs^*)/p$. The signal variance is given by $|G|^2p$. Using the fact that $s$, $v$ and $d$ are uncorrelated, the distortion variance is $\mathbb{E}(|d|^2)=\mathbb{E}(|r|^2)-|G|^2p-\sigma_v^2$. The \gls{snr}, \gls{sdr} and \gls{sndr} are thus given by
\begin{align*}
\text{SNR}&=\frac{|G|^2p}{\sigma_v^2},\ \text{SDR}=\frac{|G|^2p}{\mathbb{E}(|d|^2)},\ \text{SNDR}=\frac{|G|^2p}{\mathbb{E}(|d|^2)+\sigma_v^2},
\end{align*}
where the expectations can be evaluated using the statistics of the transmit symbols $s$.
\begin{figure}[t!]
	\centering 
	\resizebox{0.98\linewidth}{!}{%
		{\footnotesize 
\begin{tikzpicture}

\begin{axis}[
legend cell align={left},
legend style={
  fill opacity=0.8,
  draw opacity=1,
  text opacity=1,
  at={(0.03,0.03)},
  anchor=south west,
  draw=white!80!black
},
tick align=outside,
tick pos=left,
x grid style={white!69.0196078431373!black},
xlabel={Back-off \(\displaystyle p/p_{sat}\) [dB]},
xmin=-10, xmax=2,
xtick style={color=black},
y grid style={white!69.0196078431373!black},
ylabel={Magnitude [dB]},
ymin=10, ymax=30,
ytick style={color=black}
]
\fill[draw=white,shade, right color=blue!0, left color=blue!25]
        (axis cs:-10,10)--(axis cs:-6,10)--
        (axis cs:-6,30)--(axis cs:-10,30)--cycle;
\fill[draw=white,shade, right color=red!25, left color=red!0]
        (axis cs:-6,10)--(axis cs:2,10)--
        (axis cs:2,30)--(axis cs:-6,30)--cycle;
\addplot [thick, red, dotted]
table {%
-10 34.6540372081484
-9.36842105263158 32.6087150387162
-8.73684210526316 31.087991157037
-8.10526315789474 28.5407799289944
-7.47368421052632 27.0188308104076
-6.84210526315789 25.2244430197434
-6.21052631578947 23.4864896010509
-5.57894736842105 22.281926592082
-4.94736842105263 20.9516151230674
-4.31578947368421 19.4334747135741
-3.68421052631579 18.4118448547648
-3.05263157894737 17.334482002254
-2.42105263157895 16.2072014911977
-1.78947368421053 15.2499110409702
-1.15789473684211 14.4326022323849
-0.526315789473685 13.5910595496853
0.105263157894736 12.8848559366571
0.736842105263158 12.1859137160307
1.36842105263158 11.5456074599341
2 10.909302123352
};
\addlegendentry{SDR - MRT}
\addplot [thick, red, dashed]
table {%
-10 26.035446135248
-9.36842105263158 25.9469147439056
-8.73684210526316 25.9313475877458
-8.10526315789474 25.8318953975212
-7.47368421052632 25.8142600059824
-6.84210526315789 25.7111955751059
-6.21052631578947 25.6093651895551
-5.57894736842105 25.4640549900726
-4.94736842105263 25.2982471831642
-4.31578947368421 25.1514171842784
-3.68421052631579 24.9473362625768
-3.05263157894737 24.7196013827521
-2.42105263157895 24.5069131638032
-1.78947368421053 24.1926574514217
-1.15789473684211 23.8744640179319
-0.526315789473685 23.5219202367532
0.105263157894736 23.1815215415407
0.736842105263158 22.7497392424701
1.36842105263158 22.3294778265971
2 21.888514902443
};
\addlegendentry{SNR - MRT}
\addplot [thick, red]
table {%
-10 25.4761276389416
-9.36842105263158 25.0987041137553
-8.73684210526316 24.7751587891577
-8.10526315789474 23.9681690284791
-7.47368421052632 23.3646157634962
-6.84210526315789 22.4507035696087
-6.21052631578947 21.4091881880169
-5.57894736842105 20.5775381255968
-4.94736842105263 19.5921282444014
-4.31578947368421 18.4021321922284
-3.68421052631579 17.5409549526337
-3.05263157894737 16.6061228822524
-2.42105263157895 15.608082772463
-1.78947368421053 14.728502158644
-1.15789473684211 13.964865559724
-0.526315789473685 13.1708016249153
0.105263157894736 12.497076718802
0.736842105263158 11.820323936217
1.36842105263158 11.1973749511382
2 10.5758128465993
};
\addlegendentry{SNDR - MRT}
\addplot [thick, blue, dotted]
table {%
-10 38.8915094938269
-9.36842105263158 38.7251790789144
-8.73684210526316 38.3830179039907
-8.10526315789474 37.2485292515947
-7.47368421052632 35.2118015969666
-6.84210526315789 33.5516025596142
-6.21052631578947 30.8009662900729
-5.57894736842105 28.3053371809272
-4.94736842105263 26.2656099096762
-4.31578947368421 24.2894449226464
-3.68421052631579 22.6716116133152
-3.05263157894737 20.9740936148555
-2.42105263157895 19.5211783910918
-1.78947368421053 18.2988041287477
-1.15789473684211 17.0144163408321
-0.526315789473685 15.9807667206632
0.105263157894736 14.9705009279909
0.736842105263158 14.1348181271837
1.36842105263158 13.2672774680498
2 12.5248751162268
};
\addlegendentry{SDR - Z3RO}
\addplot [thick, blue, dashed]
table {%
-10 22.966432354107
-9.36842105263158 22.9518488694913
-8.73684210526316 22.9729883350053
-8.10526315789474 22.9898291201333
-7.47368421052632 22.9347533335319
-6.84210526315789 22.9336762673688
-6.21052631578947 22.9815835576264
-5.57894736842105 22.8949973060713
-4.94736842105263 22.8269088047744
-4.31578947368421 22.7913224766756
-3.68421052631579 22.6078861295905
-3.05263157894737 22.4958848984533
-2.42105263157895 22.2918014228719
-1.78947368421053 22.1274528422938
-1.15789473684211 21.86426504346
-0.526315789473685 21.6258683367107
0.105263157894736 21.3109498345807
0.736842105263158 21.0097874320717
1.36842105263158 20.637831617853
2 20.2526532571131
};
\addlegendentry{SNR - Z3RO}
\addplot [thick, blue]
table {%
-10 22.8568387040202
-9.36842105263158 22.8384086822372
-8.73684210526316 22.8497894301343
-8.10526315789474 22.8299123807711
-7.47368421052632 22.6849884418075
-6.84210526315789 22.572429357034
-6.21052631578947 22.3175055179895
-5.57894736842105 21.7967919216174
-4.94736842105263 21.2041538675991
-4.31578947368421 20.4658033226745
-3.68421052631579 19.6293320326912
-3.05263157894737 18.6583720597021
-2.42105263157895 17.6788949171639
-1.78947368421053 16.7939153512576
-1.15789473684211 15.784554932175
-0.526315789473685 14.9339246717849
0.105263157894736 14.0635138156947
0.736842105263158 13.3236721583602
1.36842105263158 12.5366662599604
2 11.8476900575586
};
\addlegendentry{SNDR - Z3RO}
\draw (axis cs:-4,27.5) node[
  scale=0.5,
  anchor=base west,
  text=black,
  rotate=0.0
]{\LARGE \textbf{Saturation regime}};
\draw (axis cs:-9.5,27.5) node[
  scale=0.5,
  anchor=base west,
  text=black,
  rotate=0.0
]{\LARGE \textbf{Linear regime}};

\end{axis}

\end{tikzpicture}}
	} 
	\caption{\gls{snr}, \gls{sdr} and \gls{sndr} of the \gls{mrt} versus \gls{z3ro} precoders for a Rapp \gls{pa} model with smoothness parameter $S=2$ and $M=64$, $M_s=4$, $M^2\beta p/\sigma_v^2=26$ dB, for a pure \gls{los} channel. The \gls{z3ro} precoder \gls{sndr} outperforms \gls{mrt} in the saturation regime.}
	\label{fig:SDR_SNR_SNDR} 
		\vspace{-1em}
\end{figure}
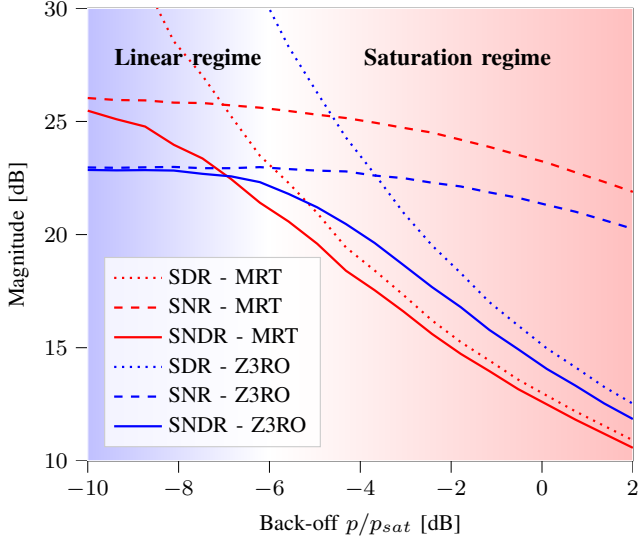
In Fig~\ref{fig:SDR_SNR_SNDR}, the performance is shown for a Rapp \gls{pa} model~\cite{dlr33776}
\begin{align*}
y_m=\frac{x_m}{\left(1+\left|\frac{x_m}{\sqrt{p_{\mathrm{sat}}}}\right|^{2S}\right)^{\frac{1}{2S}}},
\end{align*}
which does not only contain a third order distortion term. In the figure, the saturated power of the \gls{pa} $p_{\mathrm{sat}}$ is varied while the other parameters are given in the figure caption. The signal $s$ has a complex Gaussian distribution, which models an OFDM modulated signal. For low values of $p/p_{\mathrm{sat}}$, the \gls{pa} is in the linear regime and the \gls{mrt} achieves an optimal performance. The \gls{z3ro} precoder performs not as well given its reduced array gain. As the ratio $p/p_{\mathrm{sat}}$ increases, the \gls{pa} enters the saturation regime and distortion becomes non-negligible. \gls{mrt} is outperformed by the \gls{z3ro} precoder, which is only limited by distortion orders higher than three. In conclusion, the advantages of the \gls{z3ro} precoder versus \gls{mrt} can be seen in two ways, in the saturation regime:

1) For a same \gls{sndr}, the \gls{z3ro} precoder can work at a larger ratio $p/p_{\mathrm{sat}}$, implying an enhanced energy efficiency. As an example, to achieve a \gls{sndr} of 15 dB the \gls{z3ro} precoder can work with a ratio $p/p_{\mathrm{sat}}$ about 1.5 dB higher.
	
2) For a same $p/p_{\mathrm{sat}}$, the \gls{z3ro} precoder achieves a larger \gls{sndr}, implying an enhanced capacity. As an example, for $p/p_{\mathrm{sat}}=-2$ dB, the \gls{sndr} can be boosted by about 2.5 dB.


\section{Conclusion}\label{section:conclusion}

In this work, the \gls{z3ro} precoder is introduced to allow large array systems to work closer to saturation while compensating for nonlinear distortion. The design has a similar implementation complexity as \gls{mrt}, while it cancels the third-order distortion, without requiring knowledge of the \gls{pa} model and the channel statistics. Its array gain penalty is negligible in the large antenna case.

Perspectives include a novel precoder which: i) works in a multi-user setting, ii) compensates for higher order distortion terms, iii) takes into account \gls{pa} differences and iv) relaxes the zero distortion constraint.

\section{Appendix}
\label{section:appendix}


To solve problem (\ref{eq:all_real_3rd_problem}), the Lagrangian can be formed
\begin{align*}
L&=\left(\sum_{m=0}^{M-1} g_m\right)^2 - \lambda \left(\sum_{m=0}^{M-1} g_m^2-M\right)-\mu \sum_{m=0}^{M-1} g_m^3.
\end{align*}
Setting the derivative with respect to $g_m$ to zero, a quadratic equation is obtained. Hence, only two values are possible for $g_m$, that we denote by $\alpha$ and $\delta$. Consider that $M_s$ is the number of coefficients $g_m$ set to $\delta$ and $M-M_s$ are set to $\alpha$, with $M_s>0$ (otherwise the zero distortion constraint cannot be satisfied). For a fixed value of $M_s$, applying the two constraints gives
\begin{align*}
\alpha&=\frac{\sqrt{M}}{\sqrt{(M-M_s) + M_s^{1/3} (M-M_s)^{2/3}}}\\
\delta&=-\alpha\left(\frac{M-M_s}{M_s}\right)^{1/3}.
\end{align*}
Hence, setting $M_s$ values of $g_m$ to $\delta$ and $M-M_s$ to $\alpha$ give critical points of the Lagrangian. We fix $M/2>M_s$ to avoid symmetrical/equivalent solutions. 
Applying the change of variable $w_n=g_me^{\jmath \phi_n}$ leads to the first result of Theorem~\ref{theorem:3rd_order}. Let us optimize the resulting array gain with respect to $M_s$ 
\begin{align*}
\max_{x} M^2\frac{(x^{2/3}-(1-x)^{2/3})^2}{x^{1/3}+(1-x)^{1/3}},
\end{align*}
where $x=M_s/M \in \ ]0,1/2[$, with open intervals since $0<M_s<M/2$. The objective function monotonically decreases on the domain $]0,1/2[$, meaning that one has to choose the minimal $x$ and thus $M_s$ possible, \textit{i.e.}, $M_s=1$.



%

\section*{Acknowledgment}
{\ The research reported herein was partly funded by Huawei in the framework of the DASPOP project.}

\scriptsize 
\bibliographystyle{IEEEtran}
\bibliography{IEEEabrv,IEEEreferences}

\begin{thebibliography}{10}
\providecommand{\url}[1]{#1}
\csname url@samestyle\endcsname
\providecommand{\newblock}{\relax}
\providecommand{\bibinfo}[2]{#2}
\providecommand{\BIBentrySTDinterwordspacing}{\spaceskip=0pt\relax}
\providecommand{\BIBentryALTinterwordstretchfactor}{4}
\providecommand{\BIBentryALTinterwordspacing}{\spaceskip=\fontdimen2\font plus
\BIBentryALTinterwordstretchfactor\fontdimen3\font minus
  \fontdimen4\font\relax}
\providecommand{\BIBforeignlanguage}[2]{{%
\expandafter\ifx\csname l@#1\endcsname\relax
\typeout{** WARNING: IEEEtran.bst: No hyphenation pattern has been}%
\typeout{** loaded for the language `#1'. Using the pattern for}%
\typeout{** the default language instead.}%
\else
\language=\csname l@#1\endcsname
\fi
#2}}
\providecommand{\BIBdecl}{\relax}
\BIBdecl

\bibitem{auer11}
G.~Auer, V.~Giannini, C.~Desset, I.~Godor, P.~Skillermark, M.~Olsson, M.~A.
  Imran, D.~Sabella, M.~J. Gonzalez, O.~Blume, and A.~Fehske, ``How much energy
  is needed to run a wireless network?'' \emph{IEEE Wireless Communications},
  vol.~18, no.~5, pp. 40--49, 2011.

\bibitem{ENZ16}
H.~Enzinger, K.~Freiberger, and C.~Vogel, ``A joint linearity-efficiency model
  of radio frequency power amplifiers,'' in \emph{2016 IEEE International
  Symposium on Circuits and Systems (ISCAS)}, 2016, pp. 281--284.

\bibitem{lavr10}
P.~M. Lavrador, T.~R. Cunha, P.~M. Cabral, and J.~Pedro, ``The
  linearity-efficiency compromise,'' \emph{IEEE Microwave Magazine}, vol.~11,
  no.~5, pp. 44--58, 2010.

\bibitem{Fager2019}
C.~Fager, T.~Eriksson, F.~Barradas, K.~Hausmair, T.~Cunha, and J.~C. Pedro,
  ``{Linearity and Efficiency in 5G Transmitters: New Techniques for Analyzing
  Efficiency, Linearity, and Linearization in a 5G Active Antenna Transmitter
  Context},'' \emph{IEEE Microwave Magazine}, vol.~20, no.~5, pp. 35--49, 2019.

\bibitem{Larsson2018}
E.~G. Larsson and L.~Van Der~Perre, ``{Out-of-Band Radiation From Antenna
  Arrays Clarified},'' \emph{IEEE Wireless Communications Letters}, vol.~7,
  no.~4, pp. 610--613, 2018.

\bibitem{Mollen2018TWC}
C.~{Moll\'{e}n}, U.~{Gustavsson}, T.~{Eriksson}, and E.~G. {Larsson},
  ``{Spatial Characteristics of Distortion Radiated From Antenna Arrays With
  Transceiver Nonlinearities},'' \emph{IEEE Transactions on Wireless
  Communications}, vol.~17, no.~10, pp. 6663--6679, 2018.

\bibitem{cripps2006rf}
S.~C. Cripps, \emph{{RF power amplifiers for wireless communications}}.\hskip
  1em plus 0.5em minus 0.4em\relax Artech House, 2006, vol.~2.

\bibitem{aghdam2020distortion}
S.~R. Aghdam, S.~Jacobsson, U.~Gustavsson, G.~Durisi, C.~Studer, and
  T.~Eriksson, ``{Distortion-Aware Linear Precoding for Massive MIMO Downlink
  Systems with Nonlinear Power Amplifiers},'' \emph{arXiv preprint
  arXiv:2012.13337}, 2020.

\bibitem{bussgang1952crosscorrelation}
J.~J. Bussgang, ``{Crosscorrelation functions of amplitude-distorted Gaussian
  signals},'' \emph{Tech. Rep. 216, Research Lab. Electron}, 1952.

\bibitem{demir2020bussgang}
O.~T. Demir and E.~Bjornson, ``{The Bussgang Decomposition of Nonlinear
  Systems: Basic Theory and MIMO Extensions [Lecture Notes]},'' \emph{IEEE
  Signal Processing Magazine}, vol.~38, no.~1, pp. 131--136, 2021.

\bibitem{dlr33776}
C.~Rapp, ``{Effects of HPA-Nonlinearity on a 4-DPSK/OFDM-Signal for a Digital
  Sound Broadcasting System.}'' in \emph{Second European Conf. on Sat. Comm.,
  22. - 24.10.91, Liege, Belgium.}, 1991, pp. 179--184, lIDO-Berichtsjahr=1991,
  pages=6,.

\end{thebibliography}

\end{document}